\documentclass{sig-alternate}
\usepackage[utf8]{inputenc}
\usepackage[T1]{fontenc}
\usepackage{lmodern}

\usepackage{amssymb,amsmath}
\usepackage{theorem}
\usepackage{graphics}
\usepackage{amsfonts}
\usepackage{bbm,mathrsfs}
\usepackage{amscd}
\usepackage{color}
\usepackage{url}
\usepackage[plainpages=false,pdfpagelabels,colorlinks=true,citecolor=blue,hypertexnames=false]{hyperref}
\usepackage{algorithmic}
\usepackage{algorithm}
\usepackage{xspace}
\usepackage{pst-pdf}

\newcommand{\bN}{\mathbb{N}}

\newcommand{\bK}{\mathbb{K}}
\newcommand{\bZ}{\mathbb{Z}}
\newcommand{\Mat}{\mathcal{M}}
\newcommand{\PM}{\mathsf{M}}
\newcommand{\MM}{\mathsf{MM}}
\newcommand{\Id}{I}
\newcommand{\Row}[2]{#1_{#2,\ast}}
\newcommand{\Col}[2]{#1_{\ast,#2}}

\newcommand{\bigO}{{\mathcal{O}}}
\newcommand{\bigOsoft}{\tilde{\mathcal{O}}}
\newcommand{\diag}{\operatorname{diag}}

\newcommand{\CVM}{\textsf{CVM}\xspace}
\newcommand{\CVTrial}{\textsf{CVTrial}\xspace}
\newcommand{\ProbCV}{\textsf{ProbCV}\xspace}
\newcommand{\DetCV}{\textsf{DetCV}\xspace}
\newcommand{\DBZ}{\textsf{DBZ}\xspace}
\newcommand{\AZ}{\textsf{AZ}\xspace}

\newcommand{\constbal}{\textsf{BalConstr}\xspace}
\newcommand{\constnaive}{\textsf{NaiveConstr}\xspace}
\newcommand{\resolstorj}{\textsf{StorjohannSolve}\xspace}
\newcommand{\resolnaive}{\textsf{NaiveSolve}\xspace}

\newcommand{\Pivot}{\textsf{Pivot}\xspace}
\newcommand{\DBZI}{$\textsf{DBZ}^{\text{I}}$\xspace}
\newcommand{\DBZII}{$\textsf{DBZ}^{\text{II}}$\xspace}
\newcommand{\DBZIII}{$\textsf{DBZ}^{\text{III}}$\xspace}

\newcommand{\MI}{M^{\text{I}}}
\newcommand{\MII}{M^{\text{II}}}
\newcommand{\MIII}{M^{\text{III}}}
\newcommand{\PI}{P^{\text{I}}}
\newcommand{\PII}{P^{\text{II}}}
\newcommand{\PIII}{P^{\text{III}}}
\newcommand{\dI}{d_{\text{I}}}
\newcommand{\qI}{q_{\text{I}}}
\newcommand{\Inv}{\mathsf{Inv}}
\newcommand{\Rot}{\mathsf{Rot}}

\newcommand{\Vect}{\operatorname{span}}
\newcommand{\iterd}[2]{\Delta^{#1}(#2)}
\newcommand{\dec}{\hspace{0.4cm}}
\newcommand{\Elem}{E}

\newcommand{\VJoin}{\operatorname{\sf VJoin}}

\newtheorem{lemma}{Lemma}
\newtheorem{proposition}[lemma]{Proposition}
\newtheorem{theorem}[lemma]{Theorem}
\newtheorem{corollary}[lemma]{Corollary}
\newtheorem{definition}{Definition} 
\newtheorem{conjecture}{Conjecture} 

\makeatletter
\renewcommand\floatc@ruled[2]{{\@fs@cfont Algorithm} #2\par}
\makeatother

\begin{document}
\conferenceinfo{ISSAC'13,} {June 26--29, 2013, Boston, Massachusetts, USA.}
\CopyrightYear{2013}
\clubpenalty=10000 
\widowpenalty = 10000
	
\title{Complexity Estimates for Two Uncoupling Algorithms%
\titlenote{Supported in part by the MSR-INRIA Joint Centre.}}

\newfont{\authfntsmall}{phvr at 11pt}
\newfont{\eaddfntsmall}{phvr at 10pt}

\numberofauthors{3}
\author{
\alignauthor
Alin Bostan\\
\affaddr{INRIA (France)}\\
\email{\eaddfntsmall alin.bostan@inria.fr}
\alignauthor
Frédéric Chyzak\\
\affaddr{INRIA (France)}\\
\email{\eaddfntsmall frederic.chyzak@inria.fr}
\alignauthor
Élie de Panafieu\\
\affaddr{LIAFA (France)}\\
\email{\eaddfntsmall depanafieuelie@gmail.com}
}

\maketitle

\begin{abstract} Uncoupling algorithms transform a linear differential system
of first order into one or several scalar differential equations. We examine
two approaches to uncoupling: the cyclic-vector method (\CVM) and the
Danilevski-Barkatou-Z\"urcher algorithm (\DBZ). We give tight size bounds on
the scalar equations produced by \CVM, and design a fast variant of \CVM whose
complexity is quasi-optimal with respect to the output size. We exhibit a
strong structural link between \CVM and \DBZ enabling to show that, in the
generic case, \DBZ has polynomial complexity and that it produces a single
equation, strongly related to the output of \CVM. We prove that algorithm \CVM
is faster than \DBZ by almost two orders of magnitude, and provide
experimental results that validate the theoretical complexity analyses.
\end{abstract}

\vspace{1mm}
\noindent
{\bf Categories and Subject Descriptors:} \\
\noindent I.1.2 [{\bf Computing Methodologies}]: Symbolic and
Algebraic Manipulations --- \emph{Algebraic Algorithms}

\vspace{1mm}
\noindent {\bf General Terms:} Algorithms, Theory.

\vspace{1mm}

\begin{section}{Introduction}

\begin{subsection}{Motivation}

\emph{Uncoupling\/} is the transformation of a linear differential
system of first order, $Y' = M Y$,
for a square matrix~$M$ with coefficients in a rational-function
field~$\bK(X)$ of characteristic zero, into one or several
scalar differential equations $y^{(n)} = c_{n-1} y^{(n-1)} + \dots + c_0 y$,
with coefficients~$c_i$ in~$\bK(X)$. This
change of representation makes it possible to apply algorithms that
input scalar differential equations to systems.

In the present article, we examine two uncoupling algorithms: the
cyclic-vector method (\CVM) \cite{Schlesinger1908,Cope1936,ChKo02} and
the Danilevski-Barkatou-Z\"urcher algorithm (\DBZ)
\cite{Danilevski1937,Barkatou93,Zurcher94,BrPe96}.  While \CVM~always
outputs only one equivalent differential equation, \DBZ~can decompose
the system into several differential equations. This makes us consider
two scenarios: the \emph{generic case\/} corresponds to situations in
which \DBZ~does not split the system into uncoupled equations, whereas
in the \emph{general case}, several equations can be output.

For some applications, getting several differential equations is more
desirable than a single one. Besides, although the complexity of~\CVM
is rarely discussed, its output is said to be ``very complicated'' in
comparison to other uncoupling methods
\cite{Hilali83,Barkatou93,Zurcher94,Abramov99,Gerhold02}. For these
reasons, \CVM~has had bad reputation.  Because of this general belief,
uncoupling algorithms have not yet been studied from the complexity
viewpoint.  The lack and need of such an analysis is however striking
when one considers, for instance, statements like \emph{We tried to
  avoid [...] cyclic vectors, because we do not have sufficiently
  strong statements about their size or complexity\/} in a recent work
on Jacobson forms~\cite{GiesbrechtHeinle2012}.  One of our goals is a
first step towards filling this gap and rehabilitating~\CVM.

\end{subsection}

\begin{subsection}{Contribution}

In relation to the differential system~$Y' = M Y$, a
classical tool of differential algebra is the map~$\delta$, defined at
a matrix or a row vector~$u$ by $\delta(u) = u M +
u'$.  A common objective of both \CVM and \DBZ, explicitly for the
former and implicitly for the latter, as we shall prove in the present
article, is to discover a basis~$P$ of~$\bK(X)^n$ with respect to which
the matrix of the application~$\delta$ is very simple.  This matrix is
the matrix~$P[M]$ defined in \S\ref{sec:comp-mat}.
In contrast with~\CVM, which operates only on~$P$,
\DBZ~operates only on~$P[M]$ by performing pivot
manipulations without considering~$P$.  An important part of our
contribution is to provide an algebraic interpretation of the operations in \DBZ
in terms of transformations of the basis~$P$ (\S\ref{sec:DBZ_algo} and \S\ref{sec:DBZ_subalgos}).

More specifically, we analyse the degree of the outputs from \CVM
and~\DBZ, so as to compare them.  Interestingly enough, an immediate
degree analysis of (the first, generically dominating part of) \DBZ
provides us with a pessimistic exponential growth (\S\ref{sec:naive_analysis}).  We prove that this
estimate is far from tight: the degree growth is in fact only
quadratic (Theorem~\ref{prop:quadratic}).  Surprising simplifications between numerators and
denominators explain the result.  This leads to the first complexity
analysis of \DBZ.  It appears that, in contradiction
to the well-established belief, \DBZ and~\CVM have the same output on
generic input systems (Theorem~\ref{th:genericDBZ}).

With respect to complexity, another surprising contribution of our
work is that both \CVM and~\DBZ have polynomial
complexity in the generic case. Combining results, we design a fast variant of~\CVM (Theorem~\ref{th:CVTcomplexity}). Even more
surprisingly, it turns out that this fast~\CVM has better
complexity ($\approx n^{\theta + 1} d$, for $2 \leq \theta \leq 3$) than~\DBZ ($\approx n^5 d$), when applied to systems
of size~$n$ and degree~$d$. As the output size is proved to be generically
asymptotically proportional to~$n^3 d$, our improvement of~\CVM is
quasi-optimal with respect to the output size.

Another uncoupling algorithm is part of the work in~\cite{AbZi97}.  We
briefly analyse its complexity in \S\ref{sec:az} and obtain the
same complexity bound as for~\DBZ (Theorem~\ref{th:genericAZ}).

Our results remain valid for large positive characteristic.  They are
experimentally well corroborated in \S\ref{sec:implementation}.

\end{subsection}
\begin{subsection}{Previous work}

Uncoupling techniques have been known for a long time; \CVM can be traced back at
least to Schlesinger~\cite[p.~156--160]{Schlesinger1908}.
Loewy~\cite{Loewy1913,Loewy1918} was seemingly the first to prove that every
square matrix over an ordinary differential field of characteristic zero is
gauge-similar to a companion matrix.

That a linear system of first order is equivalent to \emph{several\/} scalar
linear equations is a consequence of Dickson's~\cite[p.~173--174]{Dickson1923}
and Wedderburn's~\cite[Th.~10.1]{Wedderburn1932} algorithmic results on the
Smith-like form of matrices with entries in non-commutative domains; see also~\cite[Chap
III, Sec. 11]{Poole1936}. Its refinement nowadays called
\emph{Jacobson form}~\cite[Th.~3 \& 4]{Jacobson1937} implies equivalence to
a \emph{single\/} scalar equation. This approach was further explored
in~\cite[\S6]{DwRo80} and~\cite{Dabeche79,Adjamagbo88}.

Cope~\cite[\S6]{Cope1936} rediscovered Schlesinger's \CVM, and additionally
showed that for a system of~$n$ equations over~$\bK(X)$, one can always choose
a polynomial cyclic vector of degree less than~$n$ in~$X$. A generalisation to
arbitrary differential fields was given in~\cite[\S7]{ChKo02}. The subject of
differential cyclic vectors gained a renewed interest in the 1970's, starting
from Deligne's non-effective proof~\cite[Ch.~II, \S1]{Deligne1970}. An
effective version, with no restriction on the characteristic, was given
in~\cite[\S3]{ChKo02}. \CVM has bad reputation, its output being said to
have~\emph{very complicated
coefficients\/}~\cite{Hilali83,Barkatou93,Abramov99}. However, ad-hoc
examples apart, very few degree bounds and complexity analyses are
available~\cite[\S2]{Cluzeau03}. The few complexity results we are aware
of~\cite{ChKo02,Gerhold02} only measure operations in $\bK(X)$, and do not
take degrees in~$X$ into account.

Barkatou~\cite{Barkatou93} proposed an alternative uncoupling method
reminiscent of Danilevski's algorithm for computing weak Frobenius
forms~\cite{Danilevski1937}. Danilevski's algorithm was also
generalised by Z{\"u}rcher~\cite{Zurcher94}; see
also~\cite[\S5]{BrPe96}. This is what we call~\DBZ, for Danilevski,
Barkatou, and Z{\"u}rcher. Various uncoupling algorithms are
described and analysed in~\cite{Gerhold02,Pech09}, including the already
mentioned algorithm from~\cite{AbZi97}.

\end{subsection}
\begin{subsection}{Notation and Conventions}
\medskip \noindent {\bf Algebraic Structures and Complexity.}
Let $\bK$ denote a field of characteristic zero, $\bK_d[X]$ the set of
polynomials of degree at most $d$, and $\Mat_n(S)$ and
$\Mat_{1,n}(S)$, respectively, the sets of square matrices of
dimension~$n \times n$ and of row vectors of dimension~$n$, each with
coefficients in some set~$S$.
The \emph{arithmetic size\/} of a matrix in $\Mat_n(\bK(X))$ is the
number of elements of $\bK$ in its dense representation.

We consider the \emph{arithmetic complexity\/} of algorithms, that
is, the number of operations they require in the base field~$\bK$.  For
asymptotic complexity, we employ the classical notation
$\bigO({\cdot})$, $\Omega({\cdot})$, and~$\Theta({\cdot})$, as well as
the notation $g = \bigOsoft(f)$ if there exists~$k$ such that $g/f =
\bigO(\log^k n)$.  We denote by~$\PM(d)$ (resp.\ $\MM(n,d)$) the
arithmetic complexity of the multiplication in $\bK_d[X]$ (resp.\ in
$\Mat_n(\bK_d[X])$).  When the complexity of an algorithm is expressed
in terms of $\PM$ and $\MM$, it means that the algorithm can use
multiplication (in $\bK[X]$ and $\Mat_n(\bK[X])$) as ``black boxes''.
Estimates for $\PM(d)$ and $\MM(n,d)$ are summarised in the following table,
where $\theta$~is a constant between $2$ and~$3$ that depends on the
matrix-multiplication algorithm used. For instance, $\theta=3$ for the schoolbook algorithm, and $\theta=\log_2(7)\approx 2.807$ for Strassen's algorithm~\cite{Strassen69}.
The current tightest upper bound,
due to Vassilevska Williams~\cite{VassilevskaWilliams12}, 
is $\theta < 2.3727$.
\smallskip

\centerline{%
\begin{tabular}{l|l|l|l|l}
Structure & Notation & Naive & Fast & Size\\
\hline
$\bK$ & -- & 1 & 1 & 1\\
$\bK_d[X]$ & $\PM(d)$ & $\bigO(d^2)$ & $\bigOsoft(d)$\hfill\cite{ScSt71} & $\Theta(d)$\\
$\Mat_n(\bK_d[X])$ & $\MM(n,d)$ & $\bigO(n^3 d^2)$ & $\bigOsoft( n^{\theta} d)$ \ \cite{CaKa91} & $\Theta(n^2d)$
\end{tabular}%
}
\smallskip
The complexity of an algorithm is said to be \emph{quasi-optimal} when, assuming~$\theta = 2$, it matches the arithmetic size of its output up to logarithmic factors. For instance, the algorithms in the table above have quasi-optimal complexity.

\medskip \noindent {\bf Generic Matrices.}
The notion of \emph{genericity\/} is useful to analyse algorithms on
inputs that are not ``particular''.  For parameters in some~$\bK^r$, a
property is classically said to be \emph{generic\/} if it holds out of the
zero set of a non-zero $r$-variate polynomial.
To define the notion of \emph{generic matrices}, we identify $M \in
\Mat_n(\bK_d[X])$ with the family $\{m_{i,j,k}\}$, indexed by $1\leq i,j\leq
  n,\ 0\leq k\leq d$, of its coefficients in~$\bK$.

\medskip \noindent {\bf Conventions.}
In this text, $M$~is always the input of the uncoupling algorithms. It is
assumed to be a matrix in $q(X)^{-1} \Mat_n(\bK_d[X])$ with $q(X) \in \bK_d[X]$.
It defines~$\delta$ on a matrix or a row vector~$u$ by $\delta(u) = u M + u'$.

For a matrix~$A$, we respectively denote its $i$th~row and
$j$th~column by $\Row{A}{i}$ and~${\Col{A}{j}}$.  We write
$\VJoin(r^{(1)}, \hdots, r^{(n)})$ for the matrix $A$ such that for
all~$i$, $\Row{A}{i} = r^{(i)}$.  We define the rows~$e_i$ by $\Id_n =
\VJoin(e_1,\dots,e_n)$.

A square matrix~$C$ is said to be \emph{companion\/} if beside
zero coefficients, it has~$1$s on its upper diagonal and
arbitrary coefficients on its last row, $c = (c_0, \dots, c_{n-1})$.  Thus:
\begin{equation}\label{eq:companion-matrix}
C= \VJoin(e_2, \dots, e_{n-1}, c) .
\end{equation}
We say $A$~has its $i$th~row \emph{in companion shape\/} if
$\Row{A}{i} = e_{i+1}$. 

We write $\diag(B^{(1)},\dots,B^{(n)})$ for a diagonal
block matrix given by square blocks~$B^{(i)}$.

\medskip \noindent {\bf Degrees of rational functions and matrices.}
In the present paper, the degree of a rational function is the maximum
of the degrees of its numerator and denominator.  The degree of a
vector or a matrix with rational-function coefficients is the maximum
of the degrees of its coefficients.  The following lemma expresses the
generic degrees encountered when solving a generic matrix.

\begin{lemma} \label{lem:inverse_degree}
Let $A$ be a matrix in $\Mat_n(\bK[X])$. Define $a_i$ as
$\deg(\Row{A}{i})$ and $D$ as $\sum_i a_i$.  Then, $\deg(\det(A))
\leq D$ and, for all i, $\deg(\det(A) {\Col{(A^{-1})}{i}}) \leq D -
a_i$.
When $A$ is generic with $\deg(\Row{A}{i}) = a_i$ for all $i$, those bounds are reached.
\end{lemma}

\begin{proof}
Proofs use classical techniques and are omitted.  We simply observe
that $\diag(x^{a_1},\dots,x^{a_n}) \, N$ reaches the announced bounds
when $N\in \Mat_n(\bK\smallsetminus\{0\})$~and~$\det N \neq 0$.
\end{proof}
\end{subsection}

\begin{subsection}{Companion matrices and uncoupling}\label{sec:comp-mat}

For an invertible matrix~$P$, let us perform the change of unknowns $Z
= PY$ in a system~$Y' = M Y$.  Then, $Z' =
P\, Y' + P'Y = (P\, M + P') P^{-1} Z$.
The system is therefore \emph{equivalent\/} to $Z' = P[M]\, Z$ where
$P[M]$ denotes $(P\, M + P') P^{-1}$, in the sense that the solutions
of both systems, whether meromorphic or rational, are in bijection under~$P$.

We call \emph{gauge transformation\/} of a matrix $A$ by an invertible
matrix $P$ the matrix $P[A] = (P\, A + P')P^{-1}$.  When $B = P[A]$,
we say that $A$ and $B$ are \emph{gauge-similar}.  The
gauge-similarity relation is transitive since $P[Q[A]] = (P\, Q)[A]$.
With the notation introduced above, $P[M] = \delta(P)\, P^{-1}$.

The following folklore theorem relates the solutions of a system with
the solutions of the uncoupled equations obtained from a suitable
gauge-similar system: it states that, to uncouple the system $Y' = M
Y$, it suffices to find an invertible matrix~$P$ such that $P[M]$~is
in diagonal companion block form.  This is the main motivation for
uncoupling.  We omit its proof, as it has similarity with the proofs
in \S\ref{sec:struct-thms}, and because we use no consequence of
it later in this article.  (We write $\partial f$ instead of~$f'$ for
derivations.)

\begin{theorem}
Let $P$ be an invertible matrix such that 
\begin{equation}\label{eq:diag-comp-block-decomp}
P[M] = \diag(C^{(1)}, \dots, C^{(t)})
\end{equation}
with $C^{(i)}$ companion of dimension $k_i$. 
Denote the last row of $C^{(i)}$ by
$\bigl(c^{(i)}_{0}, \dots, c^{(i)}_{k_i-1} \bigr)$.
Then, $\partial Y = M\, Y$ if and only if
\begin{equation*}
P Y = \VJoin(Z^{(1)}, \hdots, Z^{(t)})
\end{equation*}
where $Z^{(i)} = \bigl( z^{(i)}, \partial z^{(i)}, \dots,
\partial^{k_i - 1} z^{(i)} \bigr)^T$ and
\begin{equation*}
\partial^{k_i} z^{(i)} = c^{(i)}_{k_i-1} \partial^{k_i-1} z^{(i)} +
\dots + c^{(i)}_{0} z^{(i)} .
\end{equation*}
\end{theorem}

\end{subsection}
\end{section}
\begin{section}{Cyclic-vector method}

Two versions of~\CVM are available, depending on how the first row
of~$P$ is obtained: a version \ProbCV picks this first row at random,
and thus potentially fails, but with tiny probability; a deterministic
version \DetCV computes a first row in such a way that the subsequent
process provably cannot fail.  In both cases, \CVM~produces no
non-trivial diagonal companion block decomposition but only one block.

We present the randomised~\CVM only, before analysing the degree of its output
and giving a fast variant.

\begin{subsection}{Structure theorems}\label{sec:struct-thms}

Let $\iterd{k}{u}$ denote the matrix $\VJoin(u, \delta(u), \hdots,
\delta^{k-1}(u))$ of dimension~$k\times n$. The diagonal companion
block decomposition~\eqref{eq:diag-comp-block-decomp} is based on the
following folklore result.

\begin{lemma} \label{th:C0**}
Let $P$ be an invertible matrix.
Then, there exists a companion matrix~$C$ of dimension~$k$
such that
\begin{equation} \label{eq:th:C0**}
P[M] = \left( \begin{smallmatrix} C & 0\\ \ast & \ast \end{smallmatrix} \right)
\end{equation}
if and only if there exists a vector~$u$ such that
$P = \left( \begin{smallmatrix} \iterd{k}{u} \\ \ast \end{smallmatrix} \right)$
and $\delta^{k}(u) \in \Vect\bigl(u, \delta(u), \ldots, \delta^{k-1}(u)\bigr)$.
\end{lemma}

\begin{proof}
Set $\left( \begin{smallmatrix} U\\ R\end{smallmatrix} \right) := P$ 
where $U$~has $k$~rows.
Equality~\eqref{eq:th:C0**} is equivalent to
$
\delta \left( \begin{smallmatrix} U\\ R\end{smallmatrix} \right) =
\left( \begin{smallmatrix} C & 0\\ \ast & \ast \end{smallmatrix} \right)
\left( \begin{smallmatrix} U\\ R\end{smallmatrix} \right)
$,
then with
$\delta(U) = C\, U$.
This can be rewritten:
\[ 
\VJoin(\delta(\Row{U}{1}), \hdots, \delta(\Row{U}{k}))
=
\VJoin(\Row{U}{2}, \hdots, \Row{U}{k}, \Row{C}{k} U) .
\]
Set $u$ to the first row of~$U$.
This equation is satisfied if and only if
$U = \iterd{k}{u}$
and $\delta^{k}(u) \in \Vect(\iterd{k}{u})$.
\end{proof}

The following corollaries for partial companion decomposition and
diagonal companion block decomposition are proved in a very similar
fashion to the preceding lemma.  They will be used for the analysis
of \CVM and~\DBZ.

\begin{corollary} \label{th:partial_companion_decomposition}
Let $P$ be an invertible matrix.  Then,
$P[M]$ has its first $k-1$ rows in companion shape
if and only if there exists
a row vector~$u$ such that
$P = \left(\begin{smallmatrix} \iterd{k}{u}\\ \ast \end{smallmatrix}\right)$.
\end{corollary}

\begin{corollary}\label{thm:P[M]-is-diag-of-comp}
Let $P$ be an invertible matrix and
$\{C^{(i)}\}_{1 \leq i \leq t}$ a family of
companion matrices of dimension $k_i$.  Then,
\begin{equation*}
P[M] = \diag(C^{(1)}, \hdots, C^{(t)})
\end{equation*}
if and only if
there exist $t$ row vectors $\{u^{(i)}\}_{1 \leq i \leq t}$ such that
$P = \VJoin( \iterd{k_1}{u^{(1)}}, \hdots, \iterd{k_t}{u^{(t)}})$ 
and for all~$i$,
$\delta^{k_i}(u^{(i)})$ is in $\Vect(\iterd{k_i}{u^{(i)}})$.
\end{corollary}

\end{subsection}

\begin{subsection}{Classical algorithms for cyclic vectors}

The name \CVM comes from the following notion.

\begin{definition}
A \emph{cyclic vector\/} is a row vector~$u \in \bK(X)^n$ for which the matrix~$\iterd{n}{u}$ is invertible, or, equivalently, such
that the cyclic module generated by~$u$
over\/~$\bK(X)\langle\delta\rangle$ is the full vector space
$\Mat_{1,n}\bigl(\bK(X)\bigr)$ of row vectors.
\end{definition}

The next folklore method~\cite{Barkatou93,ChKo02} is justified by
Theorem~\ref{thm:CVM-generically-ok} below, which means that \CVTrial
will not fail too often.

\begin{algorithm}\small
\caption{\CVTrial: Testing if a vector is a cyclic vector}
\label{algo_rand_Cv}
\textbf{Input:} $M \in q(X)^{-1} \Mat_n(\bK_d[X])$ and $u \in \Mat_{1,n}(\bK_{n-1}[X])$ \\
\textbf{Output:} $P$, $C$ with $C$ companion and $\delta(P)P^{-1} = C$

\begin{algorithmic}[1]
\STATE set $P$ to the square zero matrix of dimension $n$
\STATE $\Row{P}{1} := u$
\STATE for $i=1$ to $n-1$, do $\Row{P}{i+1} := \delta(\Row{P}{i})$ \label{CVTrial-loop}
\STATE if $P$ is not invertible, return “Not a cyclic vector”
\STATE $C := \delta(P)P^{-1}$ \label{CVTrial-companion}
\STATE return $(P, C)$
\end{algorithmic} 
\end{algorithm}

\begin{theorem}\cite{Cope1936,ChKo02}\label{thm:CVM-generically-ok}
When $u$ is generic of degree less than~$n$, the matrix~$P =
\iterd{n}{u}$ is invertible.
\end{theorem}

\begin{proof}
It is proved in~\cite{Cope1936,ChKo02} that every
matrix~$M$ admits a cyclic vector~$u$ of degree less than~$n$.  Then,
$\det({\cdot})$~is a non-zero polynomial function of the matrix
coefficients.
\end{proof}

In \ProbCV, $u$~is chosen randomly, leading to a Las Vegas
algorithm for finding a cyclic vector.  The proof above refers to the
theorem that every matrix~$M$ admits a cyclic vector.  Churchill and
Kovacic give in \cite{ChKo02} a good survey of this subject.  They
also provide an algorithm that we denote \DetCV that takes as input a 
square matrix
$M$ and deterministically outputs a cyclic vector~$u$.  The
arithmetic complexity of this algorithm is polynomial, but worse than
that of \ProbCV.
\end{subsection}

\begin{subsection}{Degree analysis and fast algorithm}

\CVTrial computes two matrices, $P$ and~$C$, whose sizes we now
analyse.  We shall find the common bound~$\bigO(n^3 d)$.  When $u \in
\Mat_{1,n}(\bK_{n-1}[X])$ is generic, this bound is reached.  The size
$\Theta(n^3 d)$ of the output of \CVTrial is then a lower bound on the
complexity of any algorithm specifying~\CVM.  After the remark that
the complexity of the simple algorithm is above this bound, we give a
fast algorithm.

We start by bounding the degree of the matrix~$\iterd{n}{u}$.
Following the result of Churchill and Kovacic,
we make the assumption that $\deg(u)$ is less than~$n$.

\begin{lemma} \label{th:deg_delta_k}
The row vector $q^k \delta^{k}(u)$ consists of polynomials of
degree at most $\deg(u) + kd$.
\end{lemma}

\begin{proof}
The proof proceeds by induction after noting that
$q^{k+1} \delta^{k+1} = q (\delta q^k - kq^{k-1}q') \delta^k =
  (q\delta - kq') q^k \delta^k$.
\end{proof}

We list further bounds on degrees and arithmetic sizes, some of which are already in~\cite[\S2]{Cluzeau03}:

\smallskip
\begin{tabular}{l|ll}
{} & Degree & Arithmetic size\\ \hline
$P$               & $\deg(u) + (n-1) d$                      & $\bigO(n^3 d + n^2 \deg(u))$\\
$P^{-1}$          & $n\deg(u) + \frac{n(n-1)}{2} d$          & $\bigO(n^4 d + n^3 \deg(u))$\\
$\delta^{n}(u)$ & $\deg(u) + n d$                          & $\bigO(n^2 d + n \deg(u))$\\
$C$               & $n \deg(u) + \frac{n(n+1)}{2} d$         & $\bigO(n^3 d + n^2 \deg(u))$
\end{tabular}

\begin{theorem} \label{th:CVTcomplexity}
Algorithm~\CVTrial\ implements \CVM in quasi-optimal complexity
$\bigOsoft(n^{\theta + 1} d)$.
\end{theorem}

\begin{proof}
At Step~\ref{CVTrial-loop}, \CVTrial computes $\delta(\Row{P}{i}) =
\Row{P}{i} M + \Row{P}{i}'$ for successive~$i$'s.  Addition and
derivation have softly linear complexity, so we focus on the product
$\Row{P}{i} M$.  Computing it by a vector-matrix product would have
complexity $\Omega(n^2 i d)$ and the complexity of the loop at
Step~\ref{CVTrial-loop} would then be~$\Omega(n^4 d)$.  The row
$\Row{P}{i}$ has higher degree than~$M$, so the classical idea to make
it a matrix to balance the product applies.  Let $A_k \in \bK_d[X]^n$
be the rows defined by $\Row{P}{i} = \sum_{k=0}^{n-1} A_k X^{k d}$,
and $A := \VJoin(A_0, \dots, A_{n-1})$.  The product $AM$ is computed
in complexity $\bigO(\MM(n,d))$, and $\Row{P}{i} M$ is reconstructed
in linear complexity, thus performing the whole loop in complexity
$\bigO(n\, \MM(n,d))$.

Only the last row $\delta^{n}(u) P^{-1}$ of~$C$ needs to be computed
at Step~\ref{CVTrial-companion}, and the size $\Theta(n^4 d)$ of
$P^{-1}$ bans the computation of~$P^{-1}$ from any low complexity
algorithm.
$C$~can be computed in complexity
$\bigOsoft(\MM(n,nd))$.
This is achieved by solving $PY = \delta^n(u)$
by Storjohann's algorithm \cite{Storjohann03}, which
inputs $\delta^{n}(u)$ and $P$, of degree $\Theta(nd)$,
and outputs the last row $\delta^{n}(u)\, P^{-1}$ of~$C$ in
$\bigO(\MM(n,nd) \log(nd))$ operations.
\end{proof}

\end{subsection}

\end{section}

\begin{section}{The Danilevski-Barkatou-\\Z\"urcher algorithm}

We begin this section with a description of algorithm \DBZ, before a naive
analysis that gives exponential bounds.  Experiments show a
polynomial practical complexity, whence the need for a finer analysis.
To obtain it, we develop an algebraic interpretation of the algorithm.

\begin{subsection}{Description of the algorithm}\label{sec:DBZ_algo}

The input to~\DBZ is a matrix~$M \in \Mat_n(\bK(X))$;
its output $(P, C^{(1)}, \hdots, C^{(t)})$
satisfies $P[M] = \diag(C^{(1)}, \dots, C^{(t)})$ for companion
matrices~$C^{(i)}$.  \DBZ~iterates over~$P[M]$ to make it progressively
diagonal block-companion.  To do so, three sub-algorithms
are used, \DBZI, \DBZII, and \DBZIII, in order to achieve special
intermediate forms for~$P[M]$.  These forms, respectively Shape~(I),
(II), and~(III), are:
\begin{scriptsize}
\begin{equation*}
\begin{pmatrix} C & 0\\ \alpha & \beta \end{pmatrix} ,
\quad
\begin{pmatrix}
C & \begin{smallmatrix}
0 & \cdots & 0\\
\vdots & & \vdots\\
0 & \cdots & 0
\end{smallmatrix}\\
\begin{smallmatrix} v &
\begin{smallmatrix}
0 & \cdots & 0 \\ \vdots & & \vdots \\0 & \cdots & 0 
\end{smallmatrix}
\end{smallmatrix}
& \beta
\end{pmatrix} ,
\quad
\begin{pmatrix}
\ast & 
\begin{smallmatrix} 1 & 0 & \cdots & 0 
\end{smallmatrix} & 
\begin{smallmatrix} \ast & \cdots & \ast 
\end{smallmatrix} \\
\begin{smallmatrix} 0 \\ \vdots \\ 0 
\end{smallmatrix} & C & 
\begin{smallmatrix} 0 & \cdots & 0\\ \vdots & & \vdots\\ 0 & \cdots & 0 
\end{smallmatrix}\\
\begin{smallmatrix} \ast \\ \vdots \\ \ast 
\end{smallmatrix} & 
\begin{smallmatrix} 0 & \cdots & \cdots & 0\\ \vdots & & & \vdots\\ 0 & \cdots & \cdots & 0 
\end{smallmatrix} & 
\begin{smallmatrix} \ast & \cdots & \ast\\ \vdots & & \vdots\\ \ast & \cdots & \ast 
\end{smallmatrix}
\end{pmatrix} ,
\end{equation*}
\end{scriptsize}
\!\!where in each case $C$~denotes a companion matrix, $v$~is a column
vector, $\alpha,\beta$~are general matrices, and $\beta$~is square.

In the course of~\DBZ, first, \DBZI computes $\PI$ and $\MI$ such that
$\PI[M] = \MI$ has Shape~(I).  If $\MI = C$ is companion --- that is,
if $\alpha$ and~$\beta$ do not occur --- then \DBZ returns $(\PI,
\MI)$.  If not, at this point, \DBZ~has obtained a first companion
block~$C$ and we would hope that $\alpha$~is zero to apply~\DBZ
recursively to~$\beta$.  So, in general, \DBZ~tries to cancel~$\alpha$,
by appealing to~\DBZII to compute $\PII$ and~$\MII$ such that
$\PII[\MI] = \MII$ has Shape~(II).  If the obtained~$v$ is zero, then
\DBZ~can go recursively to~$\beta$.

If not, \DBZ~seems to have failed and restarts
on a matrix $\PIII[\MII] = \MIII$ with Shape~(III), which
ensures \DBZI~can treat at least one more row than previously (as
proved in~\cite{BrPe96}).  Therefore, the algorithm does not loop
forever.  The matrix~$\MIII$ on which \DBZ~starts over and the
differential change of basis~$\PIII$ associated are computed
by~\DBZIII.

\begin{algorithm}\small
\caption{\DBZ (Danilevski-Barkatou-Z\"urcher)}
\textbf{Input:} $M \in \Mat_n(\bK(X))$ \\
\textbf{Output:} $(P, C^{(1)}, \hdots, C^{(t)})$ with $C^{(i)}$ companion matrices and $P[M] = \diag(C^{(1)}, \hdots, C^{(t)})$

\begin{algorithmic}[1]
\STATE $(\PI, \MI) := \text{\DBZI}(M)$
\STATE if $\MI$ is companion then return$(\PI, \MI)$ \label{st:DBZ-fingenerique}
\STATE $(\PII, \MII) := \text{\DBZII}(\MI)$
\STATE $\left( \begin{smallmatrix} C & 0\\ \begin{smallmatrix} v & 0 \end{smallmatrix} & \beta \end{smallmatrix}\right) := \MII$ where $v \in \Mat_{n-k,1}(\bK(X))$
\STATE if $v = 0$ then
\STATE \dec  $(P, C^{(1)}, \hdots, C^{(t)}) := \text{\DBZ}(\beta)$
\STATE \dec return $(\diag(\Id_k, P) \PII \PI, C, C^{(1)}, \hdots, C^{(t)})$ \label{finbloc}
\STATE $(\PIII, \MIII) := \text{\DBZIII}(\MII)$
\STATE $(P, C^{(1)}, \hdots, C^{(t)}) := \text{\DBZ}(\MIII)$ \label{recommencer}
\STATE return $(P \PIII \PII \PI, C^{(1)}, \hdots, C^{(t)})$ \label{finrecommencer}
\end{algorithmic} 
\end{algorithm}

\end{subsection}
\begin{subsection}{Description of the sub-algorithms}\label{sec:DBZ_subalgos}

We now describe \DBZI, \DBZII, and \DBZIII in more details.  By
$\Elem_{i,j}(t)$, we denote the matrix obtained after replacing by~$t$
the $(i,j)$~coefficient in the identity matrix~$\Id_n$, and by $\Elem_{i}(u)$ the
matrix obtained after replacing the $i$th~row by the row vector~$u$ in
$\Id_n$.
Let $\Inv^{(j,n)}$ denote the matrix obtained from $\Id_n$ after
exchanging the $j$th and~$n$th rows.  Set $\Rot=\VJoin(e_n, e_1,
\hdots, e_{n-1})$.

The algorithms developed below rely on a common \Pivot subtask, which
inputs $(M, P, T) \in \Mat_n(\bK(X))^3$ with invertible $P$ and~$T$,
and outputs the update of $(M, P)$ under~$T$.  This really behaves
like a Gauge transformation, changing~$(M, P)$ to $\bigl(T[M],
TP\bigr)$.  This modification of $M$ and~$P$ only ensures the
invariant~$M = P[M_{\text{initial}}]$.

\DBZI~inputs~$M$ and outputs the tuple of matrices $(\PI, \MI)$ with
$\MI$~in Shape~(I).  It starts with $\MI = M$ and modifies its rows
one by one.  At the $i$th~iteration of the loop (Step~\ref{st:DBZI-forligne}),
the matrix~$\MI$ has its first $i-1$~rows in companion form.
To put the $i$th~row in companion form, \DBZI~sets $\MI_{i+1,i+2}$ to~1
(Step~\ref{st:DBZI-piv2}), and uses it as a pivot to cancel the other
coefficients of the row (loop at Step~\ref{st:DBZI-forcolonne}).

\begin{algorithm}\small
\caption{\DBZI}
\textbf{Input:} $M \in \Mat_n(\bK(X))$ \\
\textbf{Output:} $(\PI, \MI)$ with $\MI$ in Shape~(I) and $\PI[M] = \MI$

\begin{algorithmic}[1]
\STATE $(\PI, \MI) := (\Id, M)$
\STATE for $i=1$ to $n-1$ do \label{st:DBZI-forligne}
\STATE \dec $r := \min \bigl( \{j\ |\ \MI_{i, j} \neq 0 \ \text{and}\ j > i\} \cup \{n+1\} \bigr)$ \label{DBZI_r}
\STATE \dec if $r = n+1$ then return $(\PI,\MI)$ \label{st:DBZI-touscoeffsnuls}
\STATE \dec $(\MI, \PI) := \Pivot(\MI, \PI, \Inv^{(i+1, r)})$ \label{st:DBZI-piv1}
\STATE \dec $(\MI, \PI) := \Pivot(\MI, \PI, \Elem_{i+1,i+1}(\MI_{i,i+1})^{-1})$ \label{st:DBZI-piv2}
\STATE \dec for $j=1$ to $n$ with $j \neq i+1$ do \label{st:DBZI-forcolonne}
\STATE \dec \dec $(\MI, \PI) := \Pivot(\MI, \PI, \Elem_{i+1,j}(-\MI_{i,j}))$ \label{st:DBZI-piv3}
\STATE return $(\PI,\MI)$ \label{st:DBZI-last-return}
\end{algorithmic} 
\end{algorithm}

If $\MI_{i+1,i+2} = 0$ and 
there is a non-zero coefficient farther on the row, 
then the corresponding columns are inverted at Step~\ref{st:DBZI-piv1}
and \DBZI~goes on. 
If there is no such coefficient, the matrix has reached Shape~(I) and is returned at Step~\ref{st:DBZI-touscoeffsnuls}.

\DBZII inputs $\MI$ and outputs a tuple $(\PII, \MII)$
with $\MII$~in Shape~(II).
At Step~\ref{st:DBZII-nettoyeralpha},
it cancels the columns of the lower-left block of $\MI$ one by one,
from the last one to the second one,
using the~1's of~$C$ as pivots.
At the $\ell$th~iteration,
the lower-left block of~$\MI$ ends with $\ell$~zero columns.

\begin{algorithm}\small
\caption{\DBZII}
\textbf{Input:} $\MI$ in Shape~(I) \\
\textbf{Output:} $(\PII, \MII)$ in Shape~(II) such that $\PII[\MI] = \MII$

\begin{algorithmic}[1]
\STATE $k := $ size of the companion block of $\MI$
\STATE $(\PII, \MII) := (\Id,\MI)$ \label{st:DBZII-P-is-Id}
\STATE for $j=k$ down to 2 do \label{st:DBZII-nettoyeralpha}
\STATE \dec for $i=k+1$ to $n$ do \label{st:DBZII-loop-over-i}
\STATE \dec \dec $(\MII, \PII) := \Pivot(\MII, \PII, \Elem_{i,j-1}(- \MII_{i,j}))$ \label{st:DBZII-piv}
\STATE return $(\PII,\MII)$
\end{algorithmic} 
\end{algorithm}

\DBZIII inputs $\MII$ with $v$ non-zero and outputs the tuple $(\PIII,
\MIII)$ with $\MIII$~in Shape~(III).  The transformation of~$M$ at
Step~\ref{st:DBZIII-piv1} reverses~$v$ to put a non-zero coefficient on
the last row of~$\MII$.  Then it sets it to~1 at
Step~\ref{st:DBZIII-piv2} and uses it as a pivot to cancel the
other~$v_j$'s (Step~\ref{st:DBZIII-piv3}).  Finally, at
Step~\ref{st:DBZIII-piv4}, a cyclic permutation is applied to the rows and
columns: last row becomes first, last column becomes first.

\begin{algorithm}\small
\caption{\DBZIII}
\textbf{Input:} $\MII$ in Shape~(II), under the constraints $C$~has dimension $k$ and $v \neq 0$ \\
\textbf{Output:} $(\PIII, \MIII)$ in Shape~(III) where $\PIII[\MII] = \MIII$

\begin{algorithmic}[1]
\STATE $\MIII := \MII$
\STATE $h := \max\{i\ |\ \MIII_{i,1} \neq 0\}$ \label{st:DBZIII-h}
\STATE $(\MIII, \PIII) := \Pivot(\MIII, \PIII, \Inv^{(h,n)})$ \label{st:DBZIII-piv1}
\STATE $(\MIII, \PIII) := \Pivot(\MIII, \PIII, \Elem_{n,n}(1 / \MIII_{n,1}))$ \label{st:DBZIII-piv2}
\STATE for $i=k+1$ to $n-1$ do
\STATE \dec $(\MIII, \PIII) := \Pivot(\MIII, \PIII, \Elem_{i,n}(- \MIII_{i,1}))$ \label{st:DBZIII-piv3}
\STATE $(\MIII, \PIII) := \Pivot(\MIII, \PIII, \Rot)$ \label{st:DBZIII-piv4}
\STATE return $(\PIII, \MIII)$
\end{algorithmic} 
\end{algorithm}

\end{subsection}
\begin{subsection}{A naive degree analysis of the generic case}\label{sec:naive_analysis}

When $M$~is generic, \DBZI~outputs a companion matrix, so
\DBZ~terminates at Step~\ref{st:DBZ-fingenerique} in the generic case with
only one companion matrix in the diagonal companion block
decomposition. The proof of this fact will be given in
\S\ref{sec:alg-interp}. Therefore, the complexity of~\DBZI is
interesting in itself.

A lower bound on this complexity is the degree of its output. 
We explain here why a naive analysis of \DBZI 
only gives an exponential upper bound on this degree.\\

Let $M^{(i)}$ be the value of~$\MI$ just before
the $i$th~iteration of the loop at Step~\ref{st:DBZI-forligne}
(in particular, $M^{(1)} = M$), and $M^{(n)}$~the output value.
Remark that the matrices involved in the gauge transformations at Steps \ref{st:DBZI-piv2} and~\ref{st:DBZI-piv3}
commute with one another.
Their product is equal to 
$\Elem_{i+1}(\Row{\Inv^{(i+1,r)}[M^{(i)}]}{i})$.

\begin{lemma}
If $A \in \Mat_n(\bK_d[X])$ is a generic matrix with its first $i-1$
rows in companion form and $T = \Elem_{i+1}(\Row{A}{i})$, then $T[A]$ has
degree $3 d$.
\end{lemma}

An exponential bound on the output of~\DBZI is easily
 deduced: $\deg(M^{(n)}) \leq 3^{n-1} \deg(M)$.
We will dramatically improve this bound in the following section.

\end{subsection}
\begin{subsection}{\!\!\!Algebraic interpretation and better bounds}\label{sec:alg-interp}

To prove the announced
tight bound, it could in principle be possible to follow the same pattern as in
Bareiss' method~\cite{Bareiss68}: give an explicit form for the
coefficients of the transformed matrices~$M$, from which the degree
analysis becomes obvious.  But it proves more fruitful to find a link
between \CVM and~\DBZ, and our approach involves almost no computation.

Algorithm~\DBZ reshapes the input matrix by successive elementary gauge
transformations. It completely relies on the shape of $M$, $\MI$, $\MII$,
and~$\MIII$, while the construction of the matrices~$P$ is only a side-effect.
As illustrated in \S\ref{sec:naive_analysis}, this approach is not
well suited for degree analysis.

In this section, we focus on $P$, $\PI$, $\PII$, and~$\PIII$.
It turns out that these matrices allow nice algebraic formulations,
leading to sharp degree and complexity analyses of~\DBZ.

The following lemma, whose omitted proof is immediate from the design
of~\DBZ, provides the complexity of the computation of an elementary
gauge transformation.

\begin{lemma}\label{lem:elem-gauge-transf}
If $t \in \bK_d[X]$ and $M \in \Mat_n(\bK_d[X])$, then
the gauge transformation of $M$ by $\Elem_{i,j}(t)$ can be computed in 
$\bigO(n\, \PM(d)) = \bigOsoft(n d)$ operations in $\bK$.
\end{lemma}

\begin{subsubsection}{Analysis of\/ $\text{\DBZ}^{\text{I}}$}

We consider an execution of Algorithm~\DBZI on a matrix~$M$ of
denominator $q$, where $\deg(q)$ and $\deg(q M)$ are equal to $d$.
Let $P^{(i)}$ and~$M^{(i)}$ be the values of the matrices~$\PI$
and~$\MI$ when entering the $i$th~iteration of the loop at
Step~\ref{st:DBZI-forligne}, and $P^{(n)}$ and~$M^{(n)}$ the values
they have at Step~\ref{st:DBZI-last-return} if this step is reached.
Set~$k$ to either the last value of~$i$ before returning at
Step~\ref{st:DBZI-touscoeffsnuls} or $n$~if
Step~\ref{st:DBZI-last-return} is reached.  Consequently, the companion
block~$C$ of~$\MI$ has dimension~$k$.  The use of Algorithm \Pivot
ensures the invariant
$M^{(i)} = P^{(i)}[M]$ for all~$i\leq k$.

The following lemma gives the shape of the matrices~$P^{(i)}$.

\begin{lemma} \label{th:formePI}
For each $i$, $P^{(i)} = \VJoin(\iterd{i}{e_1}, Q^{(i)})$ for some
$Q^{(i)}$ whose rows are in~$\{e_2, \hdots, e_n\}$.
\end{lemma}

\begin{proof}
The first $i-1$ rows of~$P^{(i)}[M]$ have companion shape by design of~\DBZI.
So, by Corollary~\ref{th:partial_companion_decomposition},
there exist a vector~$u$ and a matrix~$Q^{(i)}$ with
$P^{(i)} = \VJoin(\iterd{i}{u}, Q^{(i)})$.

We now prove by induction that $\Row{P^{(i)}}{1} = e_1$ and that for
all~$a > i$, $\Row{P^{(i)}}{a} \in \{e_2, \hdots, e_n\}$.  First for
$i=1$, $P^{(1)} = \Id$, so the property holds.  Now, we assume the
property for~$P^{(i)}$ and consider the $i$th~iteration of the loop at
Step~\ref{st:DBZI-forligne}.  For~$j\neq i+1$, let $T^{(j)}$ denote the
value of the matrix involved in the gauge transformation at
Step~\ref{st:DBZI-piv3} during the $j$th~iteration of the loop at
Step~\ref{st:DBZI-forcolonne}, and let $T^{(i+1)}$~denote the matrix
used at Step~\ref{st:DBZI-piv2}.  Let $r$ denote the integer defined at
Step~\ref{DBZI_r}.

The transformations of~$P$ at Steps \ref{st:DBZI-piv1}, \ref{st:DBZI-piv2},
and~\ref{st:DBZI-piv3} imply
$P^{(i+1)} = T^{(n)} \dotsm T^{(i+2)} T^{(i)} \dotsm T^{(1)} T^{(i+1)}
\Inv^{(i+1,r)} P^{(i)}
$.
The first row of each matrix $T^{(j)}$ and each matrix~$\Inv^{(i+1,r)}$ is~$e_1$, 
so $\Row{P^{(i+1)}}{1} = \Row{P^{(i)}}{1}$
which is, by induction, equal to $e_1$.

For each integer $a > i+1$ and each~$j$, by definition 
$\Row{T^{(j)}}{a} = e_a$.
Therefore, 
$\Row{P^{(i+1)}}{a} = \Row{\Inv^{(i+1,r)}}{a} P^{(i)}$.
Moreover, if~$a=r$, then
$\Row{\Inv^{(i+1,r)}}{a}$ is equal to~$e_{i+1}$;
if not, it is equal to~$e_a$.
In both cases, by induction,
$\Row{\Inv^{(i+1,r)}}{a} P^{(i)} \in \{e_2, \hdots, e_n\}$.
\end{proof}

We are now able to give precise bounds on 
the degree of the output and the complexity of \DBZI, using Lemma~\ref{lem:inverse_degree}.

\begin{theorem}\label{prop:quadratic}
Let $k$~be the dimension of the companion block output from~\DBZI.
The degree of $q^{k-1} P^{(k)}$ is $\bigO(k d)$
and the degree of $\bigl(q^{k (k+1)/2} \det(P^{(k)})\bigr) M^{(k)}$ is~$\bigO(k^2 d)$.
\DBZI has complexity $\bigO(n^2 k\, \PM(k^2 d)) =
\bigOsoft(n^2 k^3 d)$.
\end{theorem}

It is possible to give more precise bounds on the degrees and 
even to prove that they are reached in the generic case.

\begin{proof}
By Lemmas \ref{th:formePI} and~\ref{th:deg_delta_k},
the degrees of the rows of
$\diag(1,q,\hdots, q^{i-1},1,\hdots,1) P^{(i)}$
are upper bounded by $(0,d,$ $\hdots, (i-1) d, 0, \hdots, 0)$.
Now Lemma~\ref{lem:inverse_degree} implies that
the degree of $q^{i (i-1)/2} \det(P^{(i)}) {P^{(i)}}^{-1}$
is~$\bigO(i^2 d)$.
By the invariant of \Pivot and $P[M] = \delta(P)\, P^{-1}$,
$M^{(i)} = \delta(P^{(i)}) {P^{(i)}}^{-1}$, 
we deduce that the lcm of the denominators in~$M^{(i)}$ divides
$L_i := q^{i (i+1)/2} \det(P^{(i)})$ and that
$\deg(L_i M^{(i)})$ is in $\bigO(i^2 d)$.
The degrees of the theorem follow for~$i=k$.

The computation of~$M^{(i+1)}$ from~$M^{(i)}$ uses $n$ elementary
gauge transformations on $M^{(i)}$, leading by
Lemma~\ref{lem:elem-gauge-transf} to a complexity $\bigO(n^2 \PM(i^2
d))$.  The announced complexity for~\DBZI is obtained upon summation
over~$i$ from~1 to~$k$.
\end{proof}

The output matrix $\PI = P^{(k)}$ is invertible, so $i < k$ implies
$\delta^i(e_1) \not\in \Vect(\iterd{i}{e_1})$.  Therefore, $k$~is
characterised as the least~$i \in \bN\smallsetminus\{0\}$ such that
$\delta^i(e_1) \in \Vect(\iterd{i}{e_1})$.

Informally, for random~$M$, the~$\delta^i(e_i)$ are random, so most
probably $k$~is~$n$.  Indeed, when we experiment~\DBZ on random
matrices, it always computes only one call to~\DBZI and outputs a
single companion matrix.  We make this rigourous.

\begin{theorem}\label{th:genericDBZ}
When $M$ is generic, then \DBZ has the same output as \CVM with initial vector $e_1$.
\end{theorem}

\begin{proof}
For indeterminates $q_k$ and~$m_{i,j,k}$, let $\hat M$~be the $n\times n$ matrix
whose $(i,j)$-coefficient $\hat m_{i,j}/\hat q$ has numerator
$\hat m_{i,j} = \sum_{k=0}^d \hat m_{i,j,k} X^k$
and denominator~$\hat q = \sum_{k=0}^d \hat q_k X^k$.
Replacing~$M$ by~$\hat M$ formally in
$\det(\iterd{n}{e_1})$, we obtain a polynomial
in the~$\hat q_k$'s and the~$\hat m_{i,j,k}$'s.
This polynomial is non-zero
since for $M = \VJoin(e_2, \hdots, e_n, e_1)$, $\iterd{n}{e_1} = \Id$.  This
proves that when $M$~is generic, $e_1$~is a cyclic vector
for~$M$, so Shape~(I) is reached with empty $\alpha$ and~$\beta$,
and \DBZ~behaves as~\DBZI and as~\CVM with initial vector~$e_1$.
\end{proof}

\end{subsubsection}
\begin{subsubsection}{Analysis of\/ $\text{\DBZ}^{\text{II}}$}

As for~\DBZI, a naive analysis would lead to
the conclusion that the degrees in~$\PII$ increase
exponentially during execution of~\DBZII.
We give an algebraic interpretation of~$\PII$ that permits a tighter
degree and complexity analysis of~\DBZII.

In this section, we consider the computation of~\DBZII on a
matrix~$\MI$ in Shape~(I) whose block~$C$ has dimension~$k$.  Let
$\qI$ denote the denominator of~$\MI$, and $\dI$~be a common bound on
$\deg(\qI)$ and~$\deg(\qI \MI)$.

Let $\Gamma$ (or~$\Gamma_\beta$) denote the operator on vectors or matrices with
$n-k$~rows defined by $\Gamma(v) = \beta v - v'$.  Observe that, as in
Lemma~\ref{th:deg_delta_k}, $\deg(\qI^k \Gamma^{k}(v))$ is bounded
by $\deg(v) + k \dI$.

The loop at Step~\ref{st:DBZII-nettoyeralpha} processes the~$j$'s in decreasing order.
Let $P^{(j)}$ and $M^{(j)}$ be the values of the matrices~$\PII$
and~$\MII$ just before executing the loop at
Step~\ref{st:DBZII-loop-over-i} in~\DBZII.

\begin{lemma}\label{th:formePII}
For each~$j$, the matrix~$P^{(j)}$ has the shape
\begin{equation}\label{eq:shape-of-Pj}
P^{(j)} = \begin{pmatrix} \Id & 0\\ A^{(j)} & \Id \end{pmatrix}
\end{equation}
where $A^{(j)}$ is a matrix of dimension~$(n-k)\times k$.
Furthermore, for all~$a < j$ and for~$a = k$,
$\Col{A^{(j)}}{a} = 0$
and for $j \leq a < k$,
\begin{equation} \label{eq:col_Aj}
\Col{A^{(j)}}{a} = \Gamma(\Col{A^{(j)}}{a+1}) - \Col{\alpha}{a+1}.
\end{equation}
\end{lemma}

\begin{proof}
The matrix~$P^{(k)}$ is the identity, owing to
Step~\ref{st:DBZII-P-is-Id}; for~$k<j$, $P^{(j)}$~is equal to the
product of all the matrices~$T$ previously introduced for the gauge
transformations at Step~\ref{st:DBZII-piv} for greater values of~$j$.
Each of those matrices has a block decomposition of the form
$\left( \begin{smallmatrix} \Id & 0\\ B & \Id \end{smallmatrix}
\right)$.  Therefore, their product~$P^{(j)}$ has
shape~\eqref{eq:shape-of-Pj}, where $A^{(j)}$~is the sum of the
blocks~$B$'s.  Whether $j = k$ or~$j < k$,
for~$a < j$ and for~$a = k$, and for each~$T$,
$\Col{B}{a} = 0$; therefore, ${\Col{A^{(j)}}{a}} = 0$.

Since 
$P^{(j)} = \left( 
\begin{smallmatrix} \Id & 0\\ A^{(j)} & \Id \end{smallmatrix} 
\right)$,
its inverse is
$\left( 
\begin{smallmatrix} \Id & 0\\ -A^{(j)} & \Id \end{smallmatrix} 
\right)$
\begin{equation} \label{eq:Aj_block}
M^{(j)} = P^{(j)}[\MI] =
\begin{pmatrix}
C & 0\\
A^{(j)} C + \alpha - \Gamma(A^{(j)}) & \beta
\end{pmatrix}.
\end{equation}

By the design of~\DBZII, $A^{(j)} C + \alpha - \Gamma(A^{(j)})$
ends with $k-j$ zero columns.
We consider the $(a+1)$th~column of~\eqref{eq:Aj_block}
and use the fact that $\Col{A^{(j)}}{k} = 0$,
to obtain~\eqref{eq:col_Aj}.
\end{proof}

This leads to the degree and complexity analysis of~\DBZII.

\begin{proposition} \label{th:PII}
Both\/ $\deg\bigl(\qI^{k-1} \PII\bigr)$ and\/ $\deg\bigl(\qI^k \MII\bigr)$ are in $\bigO(k \dI)$.
The complexity of~\DBZII is
$\bigO((n-k)^2 k^2 \PM(\dI)) = \bigOsoft((n-k)^2 k^2 \dI)$.
\end{proposition}

\begin{proof}
The degree of~$P^{(j)}$ is equal to the degree of~$A^{(j)}$.
Equation~\eqref{eq:col_Aj} implies, after using~$\Col{A^{(k)}}{a} = 0$,
that for each~$a$,
\[ \Col{A^{(j)}}{a} = - \sum_{i=1}^{k-a} \Gamma^{i-1}(\Col{\alpha}{a+i}). \]
Therefore,
$\deg(\qI^{k-j+1} A^{(j)})$ is $\bigO((k-j) \dI)$.
{}From Equation~\eqref{eq:Aj_block}, it follows that the degree of~$\qI^{k-j+2} M^{(j)}$
is also $\bigO((k-j) \dI)$.
For $j = 2$, we conclude that both
$\deg(\qI^{k-1} \PII)$ and $\deg(\qI^k \MII)$ are~$\bigO(k \dI)$.

The computation of~$M^{(j+1)}$ from~$M^{(j)}$ by the loop at
Step~\ref{st:DBZII-loop-over-i} involves $n-k$ elementary gauge
transformations.  Each one computes ${n-k}$ (unbalanced) multiplications of
elements of $\beta$ with elements of $M^{(j)}$.  The cost is
then $\bigO((n-k)^2 (k-j) \PM(\dI))$.  We obtain the complexity of~\DBZII
by summation over~$j$ from~$2$ to~$k$.
\end{proof}

\end{subsubsection}
\begin{subsubsection}{Analysis of\/ \DBZIII and \DBZ}

Let $(P, C^{(1)}, \hdots, C^{(t)})$ be the output of \DBZ on~$M$.
Corollary~\ref{thm:P[M]-is-diag-of-comp} states that $P =
\VJoin(\iterd{k_1}{u^{(1)}}, \hdots, \iterd{k_t}{u^{(t)}})$.  The
degrees of the matrices transformed by~\DBZ, and thus its complexity,
are obviously linked to the degrees of the vectors~$u^{(i)}$.
Focusing the analysis on the degree of~$u^{(1)}$ will result in the
exponential degree bound~$\bigO\bigl(n^{\bigO(n)} d\bigr)$, which we
believe is not pessimistic.  In turn, this seems to be a lower bound
on the complexity of \DBZ.
\begin{conjecture}
The complexity of~\DBZ is more than exponential in the worst case.
\end{conjecture}

We shall show that this explosion originates in the recursive calls at
Step~\ref{finrecommencer}.  Unfortunately, we have been unable to
exhibit a matrix~$M$ leading to an execution with more than one
recursive call, such cases being \emph{very\/} degenerate.

We now drop the exponent and write $u$ for~$u^{(1)}$.  As $u$~can only
be modified at Step~\ref{finrecommencer}, we consider the initial flow
of an execution, as long as the~$\MI$'s are not companion and
the~$v$'s are non-zero; this excludes any return at Step
\ref{st:DBZ-fingenerique} or~\ref{finbloc}.

Set $P^{(\text{I}, r)}$, $M^{(\text{I}, r)}$, $P^{(\text{II}, r)}$,
and $P^{(\text{III}, r)}$ to the values of $\PI$, $\MI$, $\PII$, and
$\PIII$ just before the $r$th~call at Step~\ref{finrecommencer}.  The
matrix~$M^{(\text{I}, r)}$ has Shape~(I) and is by construction
gauge-similar to~$M$: for some invertible~$P^{(r)}$, $P^{(r)}[M] =
M^{(\text{I}, r)}$, and, by Lemma~\ref{th:C0**}, there
exist~$u^{(r)}$, $Q^{(r)}$, and~$k_r$ such that
\[ P^{(r)} = \VJoin\bigl(\iterd{k_r}{u^{(r)}}, Q^{(r)}\bigr). \]

This leads to a new interpretation of~\DBZ:
it tests several vectors $u^{(r)}$ and 
iterates $\delta$ on them to construct the matrices~$P^{(r)}$,
until $P^{(r)}[M]$~is
companion (Step~\ref{st:DBZ-fingenerique}) or 
allows a block decomposition (Step~\ref{finbloc}).

\begin{proposition} \label{th:unew}
Write~$P^{(\text{II}, r)}$ by blocks as
$\left( \begin{smallmatrix}
\Id & 0\\
A^{(r)} & \Id
\end{smallmatrix} \right)$.
There exist an integer~$h$ and a rational function~$w$ such that
\begin{equation}
u^{(r+1)} = w 
\Row{(A^{(r)}\, \iterd{k_r}{u^{(r)}} + Q^{(r)})}{h}.
\end{equation}
\end{proposition}

\begin{proof}
By definition, $u^{(r+1)}$ is the first row of~$P^{(r+1)}$.
Step~\ref{finrecommencer} sets $P^{(r+1)} = P^{(\text{I}, r+1)}
P^{(\text{III}, r)} P^{(\text{II}, r)} P^{(r)}$.
Lemma~\ref{th:formePI} implies $\Row{P^{(\text{I},r+1)}}{1} = e_1$; in
addition, by Lemma~\ref{th:formePII}, $P^{(\text{II}, r)}$~has a block
decomposition as in the theorem statement.  So,
$u^{(r+1)} = \Row{P^{(\text{III}, r)}}{1} 
\left( \begin{smallmatrix}
\Id & 0\\
A^{(r)} & \Id
\end{smallmatrix} \right)
\left( \begin{smallmatrix}
\iterd{k_r}{u^{(r)}}\\ 
Q^{(r)}
\end{smallmatrix} \right)$.
The proof is now reduced to the existence of~$h > k_r$ such that
$\Row{P^{(\text{III}, r)}}{1} = e_h$.

In Algorithm~\DBZIII, the matrices 
involved in the gauge transformations at Step~\ref{st:DBZIII-piv3} 
commute with one another.
Let $S$ denote their product.
Set $h$ to the integer defined at Step~\ref{st:DBZIII-h}, then $h > k_r$ and
$P^{(\text{III}, r)} = \Rot\, S\, \Inv^{(h,n)}$.
By construction, 
$\Row{\Rot}{1} = e_n$,
$\Row{S}{n} = w e_n$ 
for a certain rational function $w$
defined at Step~\ref{st:DBZIII-piv2},
and $\Row{\Inv^{(h,n)}}{n} = e_h$.
This ends the proof.
\end{proof}

We now express the growth of $\deg(u^{(r)})$ with respect to~$r$.
Let $d_{\text{I}, r}$ denote the degrees
of the numerators and denominators
of $P^{(r)}[M]$, so in particular a bound for~$u^{(r)}$.
Now, Proposition~\ref{th:PII} implies that
the degree of the numerators and denominators
of $A^{(r)}$ are $\bigO(k_r d_{\text{I}, r}) = \bigO(k_r^3 d + k_r^2\deg(u^{(r)}))$.

The rational function $w$
of the theorem is the inverse of an element of
$M^{(\text{II},r)}$,
so the degree of its numerator and denominator
are $\bigO(k_r^3 d + k_r^2\deg(u^{(r)}))$.
Combined with Proposition~\ref{th:unew}, this implies
$\deg(u^{(r+1)}) = \bigO(k_r^3 d + k_r^2 \deg(u^{(r)}))$.

We could not deduce from Proposition~\ref{th:unew} any polynomial bound on
the degree of the numerator of $u^{(r)}$, but we get
$\deg(u^{(r)}) = \bigO(r d n^{2r+3} \deg(u^{(0)}))$.
The worst case of this bound is obtained when $r=n-1$.

\end{subsubsection}
\end{subsection}
\end{section}

\subsection{Link with the Abramov-Zima algorithm}\label{sec:az}

In~\cite{AbZi97}, Abramov and Zima presented
an algorithm, denoted by~$\AZ$ in the following,
that computes the solutions of inhomogeneous
linear systems~$Y' = M Y + R$
in a general Ore polynomial ring setting.
It starts by a partial uncoupling to obtain
a differential equation that cancels~$Y_1$, solves it and
injects the solutions in the initial system.
We reinterpret here its computations of a partial uncoupling, focusing
on the case of systems~$Y' = M Y$ where~$M$ is a polynomial matrix,
and we analyse the complexity in the \emph{generic case}.

\medskip
\noindent{\bf Step~1.}
  Introduce a new vector~$Z$ of dimension~$\ell \leq n$ (generically with equality),
  such that $Z_1 = Y_1$ and, for $i>1$, $Z_i$ is a linear combination
  of $Y_i,\ldots,Y_n$, such that
  where~$\beta$ is a lower-triangular matrix augmented by~$1$s on its
  upper-diagonal: $\beta_{i,i+1} = 1$ for all~$1 \leq i \leq \ell-1$.

\smallskip
\noindent{\bf Step~2.}
  Eliminate the variables~$Z_2, \hdots, Z_\ell$
  by linear combinations on the system obtained in Step~1 to get
  a differential equation of order~$\ell$ that cancels~$Y_1$.

\begin{theorem}\label{th:genericAZ}
Let~$M$ be a generic matrix of dimension~$n$
with polynomial coefficients of degree~$d$.
Then, the complexity of~$\AZ$ to uncouple the system~$Y'=MY$
is~$\bigOsoft(n^5 d)$.
\end{theorem}
\begin{proof}
When~$M$ is generic, the minimal monic differential equation that cancels~$Y_1$
has order~$n$ and its coefficients of orders~$0$ to~$n-1$ 
are the coefficients of the vector~$\delta^n(e_1) P^{-1}$,
where we have set $P=\Delta^n(e_1)$.
Thus, for generic~$M$, the integer~$\ell$ defined in Step~1
is equal to~$n$.

Step~1 implies that there is an upper-triangular matrix~$U$
such that~$U_1 = e_1$, $Z = U Y$ and~$U[M] = \beta$. 
At Step~2, the eliminations of the variables~$(Z_i)_{2 \leq i \leq n}$
are carried out by pivot operations. They transform the system $Z'= \beta Z$
into a new system~$W' = C W$ with~$W_1 = Z_1 = Y_1$ and~$C$ is
a companion matrix, which is equal to~$P[M]$.
Because of the particular shape,
the matrices matching those pivots operations
are lower-triangular with~$1$s on their diagonal.
Their product is a matrix~$L$ such that~$W = L Z$;
it is also lower-triangular with~$1$s on its diagonal.
Since~$P[M] = C = L[\beta]$ and~$\beta = U[M]$, $P=LU$.

By construction, the degree of~$P$ is~$\bigO(n d)$.
The matrices~$L$, $U$, and~$L^{-1}$ 
of its LU decomposition have degrees $\bigO(n^2 d)$~\cite{Bareiss68}.
Thus, the degree of~$\beta = U[M] = L^{-1}[P[M]]$ is~$\bigO(n^2 d)$.
Steps~1 and~2 of~\AZ compute~$\bigO(n^2)$ pivot operations,
each one involving~$\bigO(n)$ manipulations (additions and products) of polynomial coefficients
of degree~$\bigO(n^2 d)$. This leads to the announced complexity for~\AZ.
\end{proof}

It can be proved that the product $L\cdot U$ in this proof is the
LU-decomposition of~$P$; for a non-generic~$M$,
\cite{AbZi97}~implicitly obtains an LUP-decomposition.

The degree bounds in the previous proof  
are reached in our experiments:
the maximal degrees of the numerator of the matrices~$\beta$
computed for random matrices~$M$ of dimension~$n$ from~$1$ to~$6$
with polynomial coefficients of degree~$1$ are, respectively,
1, 2, 5, 10, 17, and~26.

\begin{section}{Implementation}\label{sec:implementation}

\begin{table}[b]
\begin{scriptsize}
\tabcolsep2pt
\begin{center}
\begin{tabular}{l|c|c|c|c|c|c|c}
\multicolumn{5}{c|}{}                            & $n = 100$ & $n=5$   & $n=30$\\
Algorithm & $c$ & $e$ & \multicolumn{2}{c|}{$p$} & $d = 1$   & $d=100$ & $d=30$\\
\hline
\CVM & $6.8\ 10^{-7}$ & $1.81$ & $\theta+1$ & $3.88$ & $103.11$ & $3.53$ & $155.41$\\ 
\DBZ & $7.5\ 10^{-8}$ & $1.61$ & $5$ & $6.01$ & $\infty$ & $2.3$ & $14409$\\
\constbal & $2.4\ 10^{-6}$ & $1.01$ & $\theta+1$ & $3.00$ & $12.55$ & $0.5$ & $2.7$ \\
\constnaive & $3.3\ 10^{-9}$ & $1.90$ & $4$ & $4.00$ & $1.24$ & $0.2$ & $1.64$ \\
\resolstorj & $8.2\ 10^{-7}$ & $1.75$ & $\theta+1$ & $3.87$ & $83.60$ & $3.48$ & $153.16$\\
\resolnaive & $4.8\ 10^{-8}$ & $1.52$ & $5$ & $6.22$ & $106352$ & $0.85$ & $13806$ \\
\hline
Output size & \multicolumn{4}{c|}{} & $1000300$ & $13010$ & $810960$
\end{tabular}
\end{center}
\end{scriptsize}
\vskip-12pt
\begin{small}
\caption{Experimental complexity of~\DBZ, \CVM, and their sub-algorithms; common output size match $n^3 d$}\label{tab:bivariate_complexities}
\end{small}
\end{table}

\begin{figure}[t]
\begin{scriptsize}
\tabcolsep2pt
\begin{center}
\includegraphics[scale = 0.3]{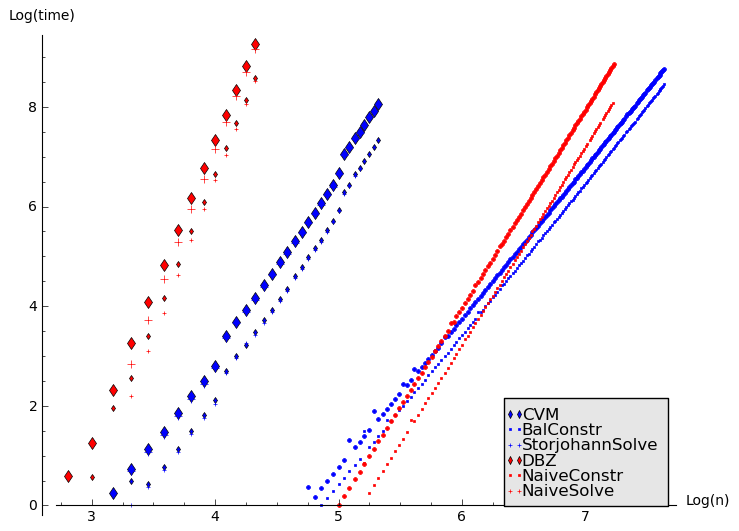}
\end{center}
\end{scriptsize}
\vskip-12pt
\begin{small}
\caption{Timings for \DBZ and \CVM on input matrices~$M$ of dimension~$n$ and coefficients with fixed degree~$d=15$ (smaller marks) or~$d=20$ (larger marks)}\label{fig:graphic}
\end{small}
\end{figure}

We have implemented the \DBZ algorithm
and several variants of the \CVM algorithm
to evaluate the pertinence of our theoretical complexity analyses and
the practical efficiency
of our algorithmic improvements.
Because of its fast implementations of
polynomial and matrix multiplications,
we chose the system Magma, using its release~V2.16-7 on Intel Xeon
5160 processors (3~GHz) and 8~GB of RAM.

Our results are summarised in Table~\ref{tab:bivariate_complexities}
and Figure~\ref{fig:graphic}.  We fed our algorithms with matrices of
size~$n$ and coefficients of degree~$d$ over~$\bZ/1048583\,\bZ$.
Linear regression on the logarithmic rescaling of the data was used to
obtain parameters $c$, $e$, and~$p$ that express the practical
complexity of the algorithms in the form~$c d^e n^p$.  For the
exponents~$p$, both theoretical and experimental values are shown for
comparison.  Sample timings for particular~$(n,d)$ are also given.
\constbal and \constnaive (resp. \resolstorj and~\resolnaive) compute
the matrix~$P$ (resp.~$C$) of Algorithm~\ref{algo_rand_Cv} with or
without the algorithmic improvements introduced in
Theorem~\ref{th:CVTcomplexity}.

In Figure~\ref{fig:graphic},
each algorithm shows two parallel straight lines,
for~$d=15$ and~20,
as was expected on a logarithmic scale.
The improved algorithms are more efficient
than their simpler counterparts
when~$n$ and~$d$ are large enough.

The theory predicts~$e = 1$ for all algorithms.  Observing different
values suggests that too low values of~$d$ have been used to reach the
asymptotic regime.  We also remark that, with respect to~$d$, \DBZ,
\constbal, and \resolstorj have slightly better practical complexity
than their respective couterparts \CVM, \constnaive, and \resolnaive.

The practical exponent~$p=3.00$ of \constbal is smaller than~$\theta+1$.
The algorithm consists of $n$~executions of a loop
that contains a constant number of scans of matrices
and matrix multiplications.
By analysing their contributions to
the complexity separately, we obtain $2.5\ 10^{-6} d^{0.97} n^{3.00}$
for the former, and $5.2\ 10^{-8} d^{1.34} n^{3.19}$ for the latter.
In the range of~$n$ we are analysing, 
the first contribution dominates because of its constant,
and its exponent~$3.00$ is the only one visible on
the experimental complexity of \constbal.

It is also visible that \resolstorj is the dominating
sub-algorithm of~\CVM.  Besides,
our implementation of \resolstorj is
limited by memory and cannot handle
matrices of dimension~$n$ over~130.
This bounds the size of the inputs
manageable by our \CVM implementation.
A native Magma implementation of Storjohann's algorithm
should improve the situation. However,
our implementation already beats the naive
matrix inversion, so that the experimental exponent~$3.88$
of \resolstorj is close to~$\theta + 1$.

The experimental exponent~$p$ of \DBZ is~$6.01$ instead of~$5$.
This may be explained by the fact that the matrix coefficients
that \DBZ handles are fractions.
Instead, in \constbal, the coefficients are polynomial:
denominators are extracted at the start
of the algorithm and reintroduced at the end.

\end{section}

\begin{section}{Conclusion}

It would be interesting to study the relevance of uncoupling applied
to system solving, and to compare this approach to direct methods.
It would also be interesting to combine \CVM and \DBZ into a hybrid
algorithm, merging speed of \CVM and generality of \DBZ.

\end{section}

\end{document}